\documentclass[a4paper,12pt]{article}

\usepackage{graphicx,dsfont}
\usepackage{t1enc,color,graphicx}
\usepackage{inputenc}

\usepackage{hyperref}

\newcommand{\Om}{\Omega}
\newcommand{\om}{\omega}

\newcommand{\Sch}{Schr\"odinger}

\newcommand{\beq}{\begin{equation}}
\newcommand{\eeq}{\end{equation}}
\newcommand{\be}{\begin{equation}}
\newcommand{\ee}{\end{equation}}
\newcommand{\bea}{\begin{eqnarray}}
\newcommand{\eea}{\end{eqnarray}}
\newcommand{\brs}{\begin{eqnarray*}}
\newcommand{\ers}{\end{eqnarray*}}
\newcommand{\ba}{\begin{array}}
\newcommand{\ea}{\end{array}}
\newcommand{\br}{\begin{eqnarray}}
\newcommand{\er}{\end{eqnarray}}

\newcommand{\bra}[1]{\left\langle #1 \right|}
\newcommand{\ket}[1]{\left| #1 \right\rangle}

\newcommand{\ie}{i.e.,\ }
\newcommand{\eg}{e.g.,\ }

\newcommand{\virg}[1]{``#1\,''}
\newcommand{\ve}[1]{{#1}}

\newcommand{\sub}[1]{_{\mathrm{#1}}}

\def\({\left(}
\def\){\right)}

\newcommand{\set}[1]{ \left\{  #1 \right\}}
\newcommand{\Hi}{\mathcal{H}}
\newcommand{\iu}{\mathrm{i}}
\newcommand{\eu}{\mathrm{e}}
\newcommand{\ph}{\varphi}

\let\oldfootnote\footnote
\renewcommand{\footnote}[1]{\oldfootnote{\  #1}}

\usepackage{color}

\usepackage{amsthm,amssymb,amsbsy,amsmath,amsfonts,amssymb,amscd,mathrsfs}

\newcommand{\gp}{}

\theoremstyle{plain}
\newtheorem{assumption}{Assumption}
\newtheorem{theorem}{Theorem}

\newtheorem{prop}[theorem]{Proposition}
\theoremstyle{definition}
\newtheorem{definition}[theorem]{Definition}

\theoremstyle{remark}
\renewcommand{\H}{{\cal H}}

\newcommand{\N}{{\mathbf N}}
\newcommand{\C}{{\mathbf{C}}}
\newcommand{\R}{{\mathbf{R}}}

\newcommand{\Z}{{\mathbf Z}}

\newcommand{\Rabi}{\mathrm{Rabi}} 
\newcommand{\Id}{\mathds{1}}
\renewcommand{\i}{\mathrm{i}\,}

\begin{document}

\title{Controllability of 
spin-boson systems
}
\date{}

\author{Ugo Boscain, Paolo Mason, Gianluca Panati, and Mario Sigalotti}

\maketitle

\begin{abstract}
In this paper we study the so-called spin-boson system, namely {a two-level system} in interaction with 
a distinguished mode of a quantized
bosonic field.  
We give a brief description of the controlled Rabi and Jaynes--Cummings 
models and we discuss their appearance in the mathematics and physics literature. 
We then study the controllability of the Rabi model when the control is an external field
acting on the bosonic part.
Applying geometric control techniques to the Galerkin approximation and
using perturbation theory to guarantee non-resonance of the spectrum of
the drift operator, we prove approximate controllability of the
system, for almost every value of the interaction parameter. 
\end{abstract}

 \section{Introduction and description of the models}

\subsection{The Rabi model {\it alias} the standard spin-boson model}

 In the Hilbert space $\Hi=L^2(\R,\C)\otimes \C^2$, we consider the \Sch\ equation
 \be\label{SE}
 \i \partial_t \psi=H_\Rabi \psi,
 \ee
where  
\begin{equation} \label{Rabi_original}
H_\Rabi=\frac \omega2 (P^2+X^2)\otimes \Id+ \frac\Om2\Id\otimes \sigma_3+g \, X\otimes \sigma_1,
\end{equation}
for some fixed positive constants $\om$ and $\Om$, and $g \in \R$, where $X$ is the multiplication operator defined by $(X\psi)(x)=x\psi(x)$, and $P=-i\partial_x$. Here and in the following we use $\sigma_1,\sigma_2,\sigma_3$ to denote the 
Pauli matrices 
$$\sigma_1=\left(\begin{array}{cc}0&1\\1&0\end{array}\right),\quad \sigma_2=\left(\begin{array}{cc}0&-\i\\ \i&0\end{array}\right),\quad \sigma_3=\left(\begin{array}{cc}1&0\\0&-1\end{array}\right).$$

Using the natural isomorphism  $\H\cong L^2(\R,\C^2)$, one rewrites   \eqref{SE} as the system
 \begin{align*}
\i \partial_t \psi_+&=\frac\omega2 (-\partial_x^2+x^2)\psi_++\frac\Omega2 \psi_+  +g x\psi_-\\
\i \partial_t \psi_-&=\frac\omega2 (-\partial_x^2+x^2)\psi_--\frac\Omega2 \psi_-  +g x\psi_+.
\end{align*}
In the physics literature, it is customary to introduce the  \emph{raising} and \emph{lowering} matrices
$$\sigma_+=\frac12(\sigma_1+\i\sigma_2)=\left(\begin{array}{cc}0&1\\0&0\end{array}\right),\qquad \sigma_-=\frac12(\sigma_1-\i\sigma_2)=\left(\begin{array}{cc}0&0\\1&0\end{array}\right).$$

Setting 
$$
\gp{a=\frac1{\sqrt{2}}(X +\i P),\qquad a^\dag=\frac1{\sqrt{2}}(X- \i P),}$$

and omitting tensors, $H_\Rabi$ appears in physics literature  as
\begin{equation} \label{Rabi}
 H_\Rabi=\omega \left(a^\dag a+\frac12
\right)+\frac \Omega2\sigma_3+ \frac{g}{\sqrt{2}} \, (a+a^\dag) (\sigma_- + \sigma_+).
\end{equation}

Modeling of spin-boson interactions was initiated by Rabi in the 30's \cite{Rabi1936, Rabi1937}.
To our knowledge, the Hamiltonian \eqref{Rabi} was first derived from a more fundamental model in the milestone paper \cite{Jaynes1963}, where also the simpler Jaynes-Cummings Hamiltonian appears (see Section \ref{Sec:JC}). 
A throughout discussion on various aspects of this model can be found in \cite{allen1975optical}. 
Recent results on its integrability have been obtained in  \cite{Braak}.

The Hamiltonian \eqref{Rabi} is ubiquitous in the literature, since it encodes one of the simplest possible couplings between a harmonic oscillator and a two-level system. The physical interpretation of the two factors appearing in the tensor product varies according to the context. We briefly outline the interpretation in cavity QED in the following digression.

\subsubsection{Digression: the Rabi model in the context of cavity QED}
\label{Sec:CQED}

While we refer to the specific literature for an exhaustive treatement of the problem \cite{Jaynes1963, Cohen1989, Raimond2001}, we briefly sketch the derivation of the Hamiltonian \eqref{Rabi} in the context of \emph{Cavity Quantum Electro Dynamics} (CQED).
In a typical CQED experiment, an atom moves across a Fabry--Perot cavity interacting with the quantized electromagnetic (EM) field of the cavity.  The atom is initially prepared in a special state, to guarantee both a strong interaction with the cavity field and a decay time (to the atomic ground state) longer than the experiment time-scale. 
The first goal is obtained by choosing a hydrogenoid atom (\eg a Rubidium atom) prepared in a \emph{circular Rydberg state} with high orbital angular momentum, and so with a large dipole moment. The second goal requires instead that transitions to the other atomic states are not resonant with the radiation from the environment, typically in the visible or infrared spectrum. Since the atomic energies are $E_n = - \frac{R\sub{Rb}}{n^2}$ (where $R\sub{Rb}$ is the Rydberg constant, corrected for the Rubidium), one select $n \in \N$ in such a way that the transition frequency $\Om\sub{at} := \frac{1}{\hbar} (E_n - E_{n-1})$ is in the range of the microwaves (in the case of Rb, this is obtained by choosing $n =51$). The cavity is designed in such a way that a normal mode of the EM field has frequency $\omega$, which is almost resonant with $\Om\sub{at}$, \ie  $|\Om\sub{at} - \om| \ll \om, \Om\sub{at}$.      

Under these conditions, the only appreciable interactions are those among the two eigenstates $\ket{n}$ and $\ket{n-1}$, and the distinguished mode of the EM field with frequency $\omega$. The linear space generated by  $\ket{n}$ and $\ket{n-1}$ corresponds to the factor $\C^2$ in $\Hi$, while the quantized EM mode yields a quantum harmonic oscillator, corresponding to the factor 
$L^2(\R,\C)$. This explains the structure of the first two terms in \eqref{Rabi}.

A deeper understanding of \eqref{Rabi} is obtained by deriving it, heuristically, from the Pauli--Fierz model, see \cite[Chapter 13]{Spohn2004}. We consider a hydrogenoid atom, as, \eg Rubidium, interacting with the quantized electromagnetic (EM) field in the cavity.  The electron, with effective mass $m_*$ and charge $- e$, is described by a form factor $\ph \in C^{\infty}_0(\R^3)$, which is assumed to be a radial positive function, normalized so that $\int_{\R^3} \ph(x) dx = 1$. 
The atomic core is supposed at rest at the origin or, equivalently, moving across the cavity with constant velocity. 

Within the Pauli--Fierz model, the quantized EM fields are described by operator-valued distributions acting on the Hilbert space 
$\Hi\sub{at} \otimes \Hi\sub{f}$, where  $\Hi\sub{at} = L^2(\R^3, dr)$ corresponds to the atom, while 
$$
\Hi\sub{f} = \bigoplus_{n \in \N} \( L^2(\R^3, dk) \otimes \C^2 \)^{\otimes n}\sub{sym}
$$
is the bosonic Fock space corresponding to the field \cite{Spohn2004}. In the Coulomb gauge, the vector potential (in free space) reads 
\begin{equation} \label{Vector potential 2}
A_{\ph}(r)  = \sum_{\lambda =1}^2  \int_{\R^3} \! {dk}  \,\, \sqrt{\frac{\hbar c^2}{2 \omega(k)}}
\, \widehat \ph(k) \,  e_{\lambda}(k) \,   
\, \(  \eu^{\iu k \cdot r} \otimes  a(k, \lambda) + \eu^{- \iu k \cdot r} \otimes  a(k, \lambda)^\dag \), 
\end{equation} 
where $\om(k) = c |k|$, $\widehat \ph$ denotes the Fourier transform of $\ph$, the three vectors $\set{k/|k|, e_1(k), e_2(k)}$ form an orthonormal basis at every point $k \in \R^3$, and  $a(k, \lambda)$ and $a(k, \lambda)^\dag$ satisfy the canonical commutations relations. Analogously, the transverse electric field reads 
\begin{equation} \label{Electric field 2}
E_{\ph, \perp}(r)  = \sum_{\lambda =1}^2  \int_{\R^3} \! {dk}  \,\,  \sqrt{\frac{\hbar \omega(k)}{2 }}
\, \widehat \ph(k) \,\, e_{\lambda}(k) \, 
\,\,  \iu \(  \eu^{\iu k \cdot r} \otimes  a(k, \lambda) - \eu^{- \iu k \cdot r} \otimes  a(k, \lambda)^\dag \).
\end{equation}
Actually, since the field is confined in the cavity,  described by a compact set $C \subset \R^3$, 
the integral appearing in  \eqref{Vector potential 2} and \eqref{Electric field 2} can be replaced by an infinite sum $\sum_{k \in \mathcal{I}}$, where  
$\mathcal{I}$ labels the solutions to the eigenvalue equation $\Delta A = \lambda_k  A$ satisfying $\nabla \cdot A = 0$ and appropriate boundary conditions on $\partial C$, see \cite[Sec. 13.6]{Spohn2004}.  

The \emph{(spinless) Pauli--Fierz Hamiltonian} is the self-adjoint operator 
\begin{equation} \label{Pauli--Fierz}
H =  \frac{1}{2 m_*}  \( \iu \gp{\hbar} \nabla_{r}  + \frac{e}{c} A_{\ph}(r) \)^2  +  \Id \otimes H\sub{f} 
- e^2 \, V_{\ph, \textrm{Coul}}(r) \otimes \Id,
\end{equation}
where 
$$
H\sub{f} =  \sum_{\lambda =1}^2  \int_{\R^3} dk  \, \hbar \om(k) \, a(k, \lambda)^\dag  a(k, \lambda)
$$
corresponds to the field energy, while 
$$
V_{\ph, \textrm{Coul}}(r) = \int_{\R^3} \!\!\! dk \,\, \frac{|\widehat \ph (k)|^2}{|k|^2} \eu^{- \iu k \cdot r} 
=  \int_{\R^3} \!\!\! dx  \int_{\R^3} \!\!\! dy  \, \frac{\ph(x - r) \ph(y)}{ 4 \pi |x-y|}
$$
is the Coulomb interaction smeared through $\ph$. 

\medskip

To derive the Rabi Hamiltonian \eqref{Rabi} from \eqref{Pauli--Fierz}, several approximations are in order.
\footnote{Obviously, in the comparison, the order of factors in the tensor product should be reversed, so that 
$L^2(\R, \C) \otimes \C^2$ is an effective reduced space for $\Hi\sub{f} \otimes \Hi \sub{at}$.
Since both conventions are well-established,  we did not dare to reverse the order of factors in any of the two models, namely in \eqref{Rabi_original} and \eqref{Pauli--Fierz}.   
}

\medskip

First, since the electron is bound to the atomic core, \virg{one loses little by evaluating the vector potential at the origin instead of at $r$, the position of the electron} \cite[Sec. 13.7]{Spohn2004}.  Within such approximation, and after a suitable change of gauge, the Hamiltonian becomes
\begin{equation} \label{Dipole Hamiltonian}
H\sub{dip} =  \( - \frac{\gp{\hbar^2}}{2 m_*} \Delta_r +  V_{\ph, \textrm{Coul}}(r) \) \otimes \Id + \Id \otimes H\sub{f}   
- e \, r \cdot E_{\ph, \perp}(0)   +  \epsilon\sub{dip}(r),
\end{equation}
where
$\ve{r} \cdot \ve{E}(0) = \sum_{j=1}^3  r_j \otimes  E_j(0)$.
The expression for $H\sub{dip}$  involves the dipole operator $D = - e \, r$, thus explaining the name \emph{dipole approximation}.
Here a quadratic potential appears, namely
$$
 \epsilon\sub{dip}(r) = \frac{1}{2} \(  \frac{2}{3}  \int_{\R^3} \!\!\! dk \,  e^2 |\widehat \ph(k)|^2 \) r^2,
$$ 
which is usually neglected in this context. 

As a second crucial approximation, { \emph{only a distinguished mode of the cavity field}} is considered as interacting with the atom. If $k_*$ is its wavevector, the integrals (or sums) above reduce to the evaluation at $k_*$. 
Accordingly,  and assuming that  $\widehat \ph(k_*) = 1$,  one has
\begin{equation} \label{?}
\begin{aligned} 
&A_{\ph}(0) & = &  \quad \sum_{\lambda =1}^2  \sqrt{\frac{\hbar c^2}{2 \omega(k_*)}}
 \,\, e_{\lambda}(k_*)   
\,\,  \left[ \Id \otimes \( a(k_*, \lambda) +  a(k_*, \lambda)^\dag \) \right]  \\
& E_{\ph, \perp}(0) & =  &  \quad \sum_{\lambda =1}^2  \sqrt{\frac{\hbar \omega(k_*)}{2}}
\,\, e_{\lambda}(k_*)  
\,\,   \left[  \Id \otimes   \iu \(   a(k_*, \lambda) -  a(k_*, \lambda)^\dag \)  \right].  \\
&H\sub{f} &=& \quad  \sum_{\lambda =1}^2   \, \hbar \om(k_*) \, a(k_*, \lambda)^\dag  a(k_*, \lambda).
\end{aligned}
\end{equation}

Neglecting the sum over the polarization vectors $e_{\lambda}$'s, we select an index $\lambda_* \in \set{1,2}$ and we set 
$\om(k_*) \equiv \om$ and $a(k_*, \lambda_*) \equiv a\sub{c}$, where the subscript $\mathrm{c}$ stands for \emph{cavity}. 
Then one has the correspondences
\begin{equation} \label{Identifications}
A_{\ph}(0)  \simeq   a\sub{c} + a\sub{c}^\dag =  \sqrt{2} X\sub{c}, \qquad \mbox{and} \qquad    
E_{\ph, \perp}(0)  \simeq  \iu (a\sub{c} - a\sub{c}^\dag) =   \gp{ -\sqrt{2} P\sub{c}},
\end{equation}
where the last equality refers to the fact that the operators 
${\footnotesize \frac{1}{\sqrt{2}}}(a\sub{c} + a\sub{c}^\dag)$  and   ${\footnotesize \frac{- \iu}{\sqrt{2}}} (a\sub{c} - a\sub{c}^\dag)$
can be identified, up to unitary equivalence, with a canonical Schr\"{o}dinger pair $(X\sub{c}, P\sub{c})$ acting on $L^2(\R, dq)$.  
Next, an \gp{inverse Fourier transform }
$\mathcal{F}^{-1}: L^2(\R, dq) \rightarrow L^2(\R, dx)$ is performed, which \gp{intertwines the 
pair $(X\sub{c}, - P\sub{c})$  with the pair  $(P, X)$,} where $P = -\iu \partial_x$ and $(X\psi)(x) = x \psi(x)$. 
Correspondingly, $\mathcal{F}^{-1}  \, a \, \sub{c} \mathcal{F} = - \iu a$, so that $H\sub{f}$ is unitarily transformed into 
$\tilde H\sub{f} = \hbar \om\sub{c} a^\dag a$, which corresponds to the first addendum in \eqref{Rabi}, up to an irrelevant constant.  Moreover, 
\begin{equation} \label{Identifications_2}
A_{\ph}(0)  \simeq  - \iu (a - a^\dag) =  \gp{\sqrt{2} P}, \qquad \mbox{and} \qquad    
E_{\ph, \perp}(0)  \simeq  a + a^\dag =   \sqrt{2} X,
\end{equation}
which shows that the interaction term (the third addendum in \eqref{Rabi_original} or in  \eqref{Rabi}) is proportional to the electric field, which is assumed to be uniform in view of the dipole approximation. 


Third, one assumes that effectively {\emph{only two atomic states}} $\ket{n-1}\equiv \ket{g}$ and $\ket{n}\equiv \ket{e}$ are appreciably coupled to the radiation field. Then, the factor $\Hi\sub{at}$ is replaced by a $2$-dimensional Hilbert space
$$
\Hi_{2L} =  \mathrm{Span}_{\C}\set{\ket{g}, \ket{e}} \simeq \C^2,
$$
and the atomic Hamiltonian (the first term in \eqref{Dipole Hamiltonian}) by the 2-level Hamiltonian $\frac{\hbar \Om\sub{at}}{2} \sigma_3$, where $\hbar \Om\sub{at} = E_n - E_{n-1}$. 

Finally, if the two selected atomic states are {\emph{circular Rydberg states}} one has that the dipole operator 
 is off-diagonal in the basis $\set{\ket{g}, \ket{e}}$, namely
$$
\bra{g} D \ket{g} =0, \qquad   \bra{e} D \ket{e} =0.   
$$
Moreover, there exist $d > 0$ and a versor $e\sub{at} \in \C^3$ such that $d \, e\sub{at} = \bra{e} D \ket{g}$.
We identify a state $\psi\sub{at} \in \Hi_{2L}$ with the pair $(\psi_g, \psi_e) \in \C^2$ via 
$\psi\sub{at} =  \psi_g  \ket{g} +  \psi_e \ket{e}$, so that  
$$
D \psi\sub{at} = d \( \psi_g \, e\sub{at}^* \ket{e} + \psi_e \, e\sub{at} \ket{g} \).  
$$
By using the latter fact, one obtains
\begin{equation} \label{Final expression}
D \cdot  E_{\ph, \perp}(0) =  \sqrt{\frac{\hbar \omega(k_*)}{2}} d
\,\, e_{\lambda_*}(k_*) \cdot 
\,\,   \left[   e\sub{at}^* \(   \sigma_+ \otimes \iu  (a\sub{c} - a\sub{c}^{\dag}) \)  
+ e\sub{at} \(\sigma_{-} \otimes  \iu (a\sub{c} - a\sub{c}^{\dag}) \)
 \right],  \\
\end{equation}
where, with a little abuse of notation, we wrote $\sigma_{-} = \ket{g}\bra{e}$ and $\sigma_{+} = \ket{e}\bra{g}$.
When adapting the previous formula to concrete physical situations, by specializing to the actual values of $e\sub{at}$ and $e_{\lambda}(k_*)$, one obtains some terms analogous to the ones appearing in \eqref{Rabi}, possibly with different weights. 
In particular, if $e\sub{at}^*=e\sub{at}$, one obtains a coupling term proportional to $\sigma_1 \otimes P\sub{c}$, which, after the Fourier transform mentioned above, becomes $\sigma_1 \otimes X$ and so agrees with the third addendum in \eqref{Rabi_original} or in \eqref{Rabi}. 

This concludes the heuristic \virg{derivation} of the Rabi Hamiltonian, starting from a fundamental model as the one given by 
\eqref{Pauli--Fierz}.  On the other hand,  a rigorous derivation of the Hamiltonian \eqref{Rabi} from a mathematical model of non-relativistic QED, as, \eg the \emph{Pauli--Fierz model}, in a suitable scaling limit, is to our knowledge absent from the literature. 

\subsection{The Jaynes--Cummings model}
\label{Sec:JC}
 
The Rabi model presented in the previous section has been considered difficult to treat in view of
the widespread opinion that  it is not explicitly solvable (cf. \cite{Braak}).

For this reason Jaynes and Cummings \cite{Jaynes1963} proposed an approximated Hamiltonian, which is obtained from $H_\Rabi$ by neglecting the so-called counter-rotating terms. More precisely, the interaction term in  $H_\Rabi$, namely 
$$ \frac{g}{\sqrt{2}} \, (a^\dag+a) \sigma_1= \frac{g}{\sqrt{2}} \, (a^\dag \sigma_- + a \sigma_+ + a^\dag \sigma_+ + a \sigma_- ),$$
is replaced by a new interaction term where the last two terms in the right-hand side are neglected 
under the assumption that 
\begin{equation} \label{Almost-resonance}
|\Omega-\omega|\ll \Omega,\qquad g\ll \Omega,\omega.
\end{equation}

According to this procedure, one obtains
\begin{align}\label{JC Hamiltonian}
H_{\mathrm{JC}}&=\omega \left(a^\dag a+\frac12
\right)+\frac \Omega2\sigma_3+ \frac{g}{\sqrt{2}} \, (a^\dag \sigma_- +  a \sigma_+)\\
&=\frac \omega2 (P^2+X^2)\otimes \Id+ \Id  \otimes \frac{\Om}{2} \sigma_3
+ \frac{g}{2} (X\otimes \sigma_1 - P\otimes \sigma_2). \nonumber
\end{align}

\noindent This model is considerably simpler than the Rabi one, since it admits the conserved quantity 
$$ C=a^\dag a\otimes \Id + \Id\otimes \left(\begin{array}{cc}1&0\\0&0\end{array}\right),$$
which represents the total number of excitations. 
Indeed, the interaction term in $H_{\mathrm{JC}}$ contains two parts: the term $a^\dag \sigma_-$ creates an excitation of the bosonic mode and destroys one of the two-level system, while $a\sigma_+$ acts in the opposite way.

\medskip

The heuristic justification of the approximation leading to \eqref{JC Hamiltonian} is based on separation of time scales. 
Indeed, when considering the dynamics in the interaction picture, the terms  $a^\dag \sigma_-$  and $a \sigma_+$, which conserve the total number of excitations of the system, evolve periodically with frequency $|\Om - \om|$, while the remaining terms $a^\dag \sigma_+$ and  $a \sigma_- $ oscillate with frequency $\Om + \om$.  If the \emph{almost-resonance condition} \eqref{Almost-resonance} is satisfied, the latter oscillations average to zero on the longer time scale $|\Om - \om|^{-1}$.  

\noindent On the other hand, a rigorous mathematical justification for this approximation seems absent from the literature. However, as already noticed by  Rouchon in \cite{rouchon-st-louis}, it seems possible to obtain it by adapting the methods developed in \cite{Panati2002, PanatiSpohnTeufel2003, Teufel}. A similar task, by analytical methods, has been accomplished by Ervedoza and Puel on a related model \cite{Ervedoza2009b}.

Notice that in applications to circuits QED the assumption $g\ll \Omega,\omega$ is in general not satisfied, \gp{since it is possible to achieve spin-filed interactions which are not much smaller than the field energy \cite{Schusteretal}}. This is a motivation for the direct study of the Rabi model, which is pursued here.  

\subsection{Modeling an external control field}

In a wide variety of experimental situations one can act on the system by an external field. The goal of the controller could be to lead the system from a given initial state to a prescribed final one. 
For spin-boson models, this amounts to study the control problem 
\begin{equation}
\i \partial_t\psi(t)=H_0\psi(t) +H\sub{c}(u(t))\psi(t),
\label{controlled}
\end{equation}
where $H_0$ represents either the Rabi or the Jaynes--Cummings Hamiltonian and 
$H\sub{c}$ is a self-adjoint  operator describing the coupling between the system and the controlled external field. 
The operator $H\sub{c}$ depends on the control function $u$ which, in general,  takes values in $\C^m$.  

In most cases the external field can act on the bosonic mode only, while the spin mode is not directly accessible.
This leads to  a control Hamiltonian of the form
\begin{equation}
H\sub{c}(u(t))=h\sub{c}(u(t))\otimes\Id.
\label{control-ham}
\end{equation}
where $h\sub{c}(u(t))$ is a self-adjoint operator in $L^2(\R)$.
One of the simplest form for the operator  $h\sub{c}(u(t))$ is the following 
\begin{equation}
h\sub{c}(u(t))=u(t)X.
\label{linear-cont}
\end{equation}

As argued in Section \ref{Sec:CQED},  in the context of cavity QED and within the dipole approximation, the operator 
$X$ corresponds to the electric field at the origin, see  \eqref{Identifications_2}.  Therefore, in this context, the control \eqref{linear-cont} corresponds to an external EM field which rescales the value of $E_{\ph, \perp}(0)$ by a time-dependent factor $u(t)$. Notice that the linearity in $X$ is a consequence of the \emph{dipole approximation}. 

Simmetrically, one might also consider a control consisting of an external field which modulates the value of $A_{\ph}(0)$, as \eg an externally-controlled magnetic field. In view of \eqref{Identifications_2}, this amounts to consider a control term $\widetilde{h}\sub{c}(u(t))=u(t) P$, a problem which will be investigated elsewhere.

\subsection{Other related controlled models} 

Controllability of finite-dimensional approximations of spin-boson systems have first been obtained in the physics literature, via constructive methods, in \cite{Law1996} and \cite{Kneer1998}. 
The results have been extended to the full infinite-dimensional model in \cite{Yuan2007} via analysis of non-resonances and in \cite{Boscain_Adami} by adiabatic arguments (see also \cite{adiabatiko}).

A mathematical description of the controllability properties of a related model is given in \cite{Ervedoza2009b}.  The main technical tool is an explicit estimate of the approximation error with respect to the Lamb--Dicke limit. 

For the counterpart of the Eberly--Law model for more than one trapped ion an approximate controllability result is obtained  in \cite{Bloch2010} (see also \cite{Rangan2004a}), 
based on  the analysis on a sequence of nested finite-dimensional systems. In the recent papers \cite{Keyl2014,PaduroSigalotti} the (approximate) controllability of the system is established by considering  different families of controlled dynamics.

\section{Approximate controllability of the controlled Rabi model}

We consider here the approximate controllability problem for a system of the form \eqref{controlled}, where $H_0$ is the Rabi Hamiltonian and $H\sub{c}$ takes the form \eqref{control-ham}-\eqref{linear-cont}. 

The problem under consideration is then
\begin{equation}\label{control-problem}
\left\{
\begin{array}{l}
i\partial_t\psi=H_\Rabi\psi+u (X\otimes \Id) \psi,\\
u\in [0,\delta],\\
|\Omega-\omega|\ll \Omega,\\
\delta,\Omega>0.
\end{array}
\right.
\end{equation}

The main result of the paper is the following. 
The precise definition of approximate controllability will be given in the next section. 
It basically means that, for every choice of the initial and final state, there exists an admissible control 
law $u$, depending on the time, which steers
the initial state arbitrarily close to the final one.

\begin{theorem}\label{Tmain}
System \eqref{control-problem} 
is approximately controllable for almost every $g\in \R$.
\end{theorem}

The proof of Theorem~\ref{Tmain} goes by studying the applicability of a general approximate controllability result in dependence on the parameter $g$.
In order to do so, we 
 have to use perturbation theory in the parameter $g$ up to order 4.

\section{An approximate controllability result}
\label{s-recall}

We are going to recall a general controllability result for bilinear quantum systems in an abstract setting.

In a separable Hilbert space ${H}$, endowed with the Hermitian product $\langle \cdot,\cdot\rangle$, we consider the following control system
\begin{equation}\label{EQ_main}
\frac{d}{dt} \psi=(A+u(t) B) \psi,\quad u(t)\in U,
\end{equation}
where $(A,B,U)$ satisfies
the following
assumption.

\goodbreak 

\begin{assumption}\label{ASS_1}
{$U$ is a subset of $\R$ and} $(A,B)$ is a pair of (possibly unbounded) linear operators in $H$ such that
\begin{enumerate}
\item $A$ is skew-adjoint on its domain $D(A)$;
\item there exists a Hilbert basis $(\phi_k)_{k\in \mathbf{N}}$ of $H$ consisting of eigenvectors of $A$: for every $k$, $A \phi_k =\mathrm{i} \lambda_k \phi_k$ with $\lambda_k$ in $\mathbf{R}$;
\item for every $j$ in $\mathbf{N}$, $\phi_j$ is in the domain $D(B)$  of $B$;
\item  $A+u B $ is essentially skew-adjoint for every $u\in U$; 
\item $\langle B \phi_j,\phi_k\rangle=0$ for every $j,k$ in $\mathbf{N}$ such that $\lambda_j=\lambda_k$ and $j\neq k$.
\end{enumerate}
\end{assumption}

{If $(A,B,U)$ satisfies Assumption \ref{ASS_1}, then} $A+uB$ generates a unitary group 
 $t\mapsto e^{t(A+uB)}$. By concatenation, one can define the solution of (\ref{EQ_main}) for every piecewise constant {function} $u$ taking values in $U$, for every initial condition $\psi_0$ given at time $t_0$. We denote this solution {by} $t\mapsto \Upsilon^u_{t,t_0} \psi_0$.

For every $j,k\in\mathbf{N}$,  we denote $b_{jk}=\langle \phi_j, B \phi_k \rangle$.
A pair $(j,k)$ in $\mathbf{N}^2$ is a {\emph{non-resonant transition}} of $(A,B)$ if 
$b_{jk}  \neq 0$  and,  for every $l,m$,
$|\lambda_j-\lambda_k|=|\lambda_l-\lambda_m|$ implies $\{j,k\}=\{l,m\}$ or $\{l,m\}\cap\{j,k\}=\emptyset$.

A subset $S$ of $\mathbf{N}^2$ is a \emph{chain of connectedness} of $(A,B)$ if  for every $j,k$ in $\mathbf{N}$, 
there exists a finite sequence $p_1=j,p_2,\ldots,p_r=k$   for which {$(p_l, p_{l+1}) \in S$ for every $l$} and
$\langle \phi_{p_{l+1}}, B \phi_{p_l} \rangle \neq 0$ for every $l=1,\ldots,r-1$. A chain of connectedness $S$ 
of $(A,B)$ is {\emph{non-resonant}} if every $(j,k)$ in $S$ is a {non-resonant } tran\-sition of $(A,B)$.

\begin{figure}
\begin{center}
\includegraphics[width=6cm]{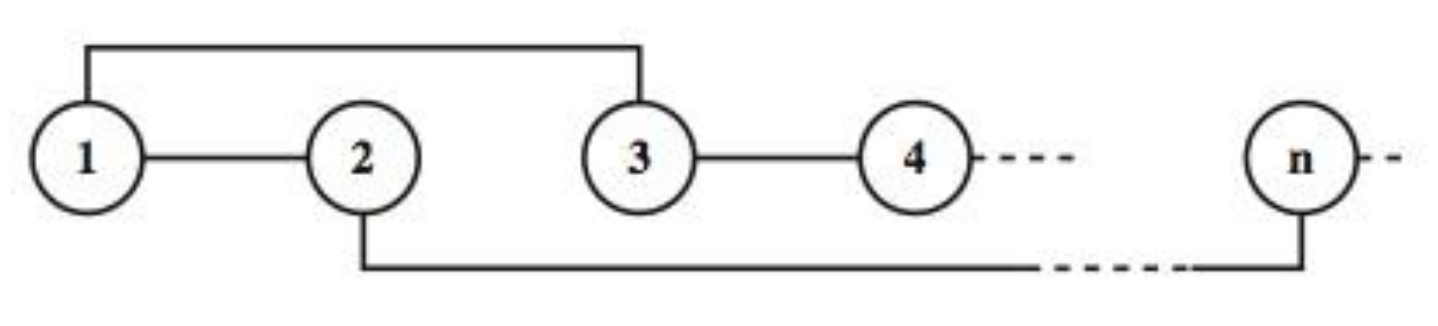}
\caption{\small Each vertex of the graph represents an eigenstate of $A$ (when the
spectrum is not simple, several nodes may be attached to the same
eigenvalue). An edge links two vertices if and only if $B$ connects
the corresponding eigenstates. In this example, $\langle \phi_1,
 B \phi_2 \rangle $ and  $\langle\phi_1, B \phi_3 \rangle $ are not
zero, while  $\langle \phi_1, B  \phi_4 \rangle =\langle \phi_2,
B  \phi_3 \rangle=\langle \phi_2, B  \phi_4 \rangle=0$.}
\label{fig:graphe}
\end{center}
\end{figure}

\medskip

\begin{definition}
Let $(A,B,U)$ satisfy Assumption \ref{ASS_1}.
We say that \eqref{EQ_main} is \emph{approximately controllable} if for every $\varepsilon>0$, for every
$\psi_0,\psi_1\in H$,
there exists a piecewise constant function $u_\varepsilon:[0,T_\varepsilon]\to U$ such that
$\|\Upsilon^u_{T_\varepsilon,0}\psi_0-\psi_1\|<\varepsilon$.
\end{definition}

\medskip

\begin{theorem}[\cite{Schrod2}]\label{pro_approx_contr_norme_H}
{Assume that $[0,\delta]\subset U$ for some $\delta>0$ and let} $(A,B,U)$ satisfy Assumption \ref{ASS_1} and admit a {non-resonant} chain of connectedness. Then
system \eqref{EQ_main} is
approximately controllable.
\end{theorem}

\section{
Proof of Theorem~\ref{Tmain}}\label{s-Tmain}

In this section, we consider the approximate controllability problem for \eqref{control-problem} and we prove 
Theorem~\ref{Tmain}. The proof of the theorem is based on a suitable application of
Theorem~\ref{pro_approx_contr_norme_H}.   

\medskip


The strategy of the proof is the following. We first show in Subsection \ref{Sec_nonresonant} that,
for almost every $g$ in $\R$, some relevant 
pairs of eigenvalues of  $H_\Rabi$ satisfy the non-resonance condition, see \eqref{ijkl}. 
This goal is \gp{achieved} by exploiting the analyticity of the eigenvalues and by using perturbation theory.
Then, in Subsection \ref{Sec_coupling}, we prove that these pairs of eigenvalues correspond to {non-resonant}
 transitions, according to the definition above. \gp{The degenerate case $\omega = \Omega$ is treated separately 
 in Section \ref{case=}}
 

\medskip


Preliminarily, we  introduce some additional notations.
Denote by $H_{\Rabi,0}$ the Hamiltonian $H_{\Rabi}$ where we set $g=0$.
The general form of $H_\Rabi$ is then
\begin{equation}\label{RabiV}
H_\Rabi=H_{\Rabi,0}+g V,
\end{equation}
where $V$ is the operator $X\otimes \sigma_1$. 

Let $(\varphi_j)_{j \in\N}$ be the standard Hilbert basis of $L^2(\R,\C)$
given by real eigenfunctions of $-\partial_x ^2+x^2$, so that
$(-\partial_x ^2+x^2)\varphi_j=(2j+1)\varphi_j$ and $\int_\R x \varphi_j(x)\varphi_{j+1}(x)dx=\sqrt{(j+1)/2}$ for $j \geq 0$.

Let $\nu_1=(1,0)^T$, $\nu_{-1}=(0,1)^T$ be the canonical orthonormal basis of $\C^2$. 
Based on $(\varphi_j)_{j\in\N}$, we obtain a Hilbert basis 
of factorized eigenstates $\Phi_{j,s}=\varphi_j \otimes \nu_s$, $j\in\N$, $s\in\{-1,1\}$,
of $H_{\Rabi,0}$ whose corresponding eigenvalues are
\begin{equation}\label{Ejs}
E_{j,s}=\omega\left(j+\frac12\right)+\frac{s}2\Omega.
\end{equation}
If $\Omega$ is not an integer multiple  of $\omega$ then each eigenvalue $E_{j,s}$ is simple. (See Figure~\ref{e-vals}.)
In the following, for ease of notations, we write in bold the elements of $\N\times\{-1,1\}$, and for every  ${\bf j}\in \N\times\{-1,1\}$ we define $n({\bf j}),s({\bf j})$ in such a way that  ${\bf j}=(n({\bf j}),s({\bf j}))$.

For $g\in \R$, denote by $E_{\bf j}^g$, ${\bf j}\in \N\times \{-1,1\}$, the eigenvalues of
$H_\Rabi$ repeated according to their multiplicities, and by
$\Phi_{\bf j}^g$, ${\bf j}\in \N\times \{-1,1\}$, an orthonormal basis of corresponding eigenstates.
By a consequence of perturbation theory, detailed in the Appendix,  we can assume that, up to the choice of a suitable labelling, each map $\R\ni g\mapsto (E_{{\bf j}}^{g},\Phi_{{\bf j}}^{g})\in \R\times L^2(\R,\C)$ is analytic, with 
$E_{{\bf j}}^0=E_{{\bf j}}$, where $E_{{\bf j}}$ is defined in \eqref{Ejs}. 
When $E_{{\bf j}}$ is simple 
we can assume, in addition, that 
$\Phi_{{\bf j}}^0=\Phi_{{\bf j}}$.

Under the assumption $|\Omega-\omega|\ll \Omega$, the only case in which $H_{\Rabi,0}$ has multiple eigenvalues is when $\omega=\Omega$. 
If this is the case all eigenvalues different from the lowest one are double.  
Equation \eqref{splitting}  allows to identify the splitting of the two-dimensional eigenspaces 
induced by the perturbation $g\mapsto H_{\Rabi,0}+g V$, as detailed in Section~\ref{case=}.

In order to study the first and higher-order derivatives  of $g\mapsto E_{j,s}^g$ at $g=0$, it is useful to introduce the quantities
\begin{align}
V_{{\bf i},{\bf j}}& = \langle \Phi_{\bf i}, (x\otimes \sigma_1) \Phi_{\bf j}\rangle=\left(\delta_{n({\bf i}),n({\bf j})-1}\sqrt{\frac{n({\bf j})}2}
+\delta_{n({\bf i}),n({\bf j})+1}\sqrt{\frac{n({\bf j}) +1 }2}
\right) (1-\delta_{s({\bf i}),s({\bf j})}). \label{interaction-matrix}
\end{align}


\subsection{Case $\omega\ne \Omega$}

\subsubsection{Step I: Relevant eigenvalue pairs  are non-resonant.}
\label{Sec_nonresonant}

Let us first prove that for almost every $g\in \R$ and every ${\bf i,j,k,l}\in \N\times\{-1,1\}$, with $({\bf i,j})\neq ({\bf k,l})$ and ${\bf i\neq j}$, 
one has $E_{\bf i}^g-E_{\bf j}^g\neq E_{\bf k}^g-E_{\bf l}^g$. In order to do so, {we} observe that it is enough to show that for fixed ${\bf i,j,k,l}\in \N\times\{-1,1\}$ as before, the set
\begin{equation}\label{ijkl}
S_{\bf i,j,k,l}=\{g\mid E_{\bf i}^g-E_{\bf j}^g\neq E_{\bf k}^g-E_{\bf l}^g \}
\end{equation}
 is of full measure. By the analytic dependence on $g$ of the eigenvalues of $H_\Rabi$, this is equivalent to say that 
 $g\mapsto E_{\bf i}^g-E_{\bf j}^g$ and $g\mapsto E_{\bf k}^g-E_{\bf l}^g$ have different Taylor expansions 
 at $g=0$.

Let us consider the Taylor expansion
$$E_{\bf j}^g = E_{\bf j}+\sum_{m=1}^{\infty} g^m E_{\bf j}^{(m)}.$$
The computation of the  coefficients $E_{\bf j}^{(m)}$
carried on
below
is based on the Rayleigh--Schr\"odinger series
(see, for instance, \cite[Chapter XII]{reed_simon}).

First of all we observe that $E_{\bf i} - E_{\bf j} = E_{\bf k} - E_{\bf l}$ is equivalent to $n({\bf i})-n({\bf j}) = n({\bf k})-n({\bf l})$ and $s({\bf i})-s({\bf j}) = s({\bf k})-s({\bf l})$, in the case in which $\Omega$ is not an integer multiple of $\omega/2$. If $\Omega=(2m+1)\omega/2$ for some non-negative integer $m$, then $E_{\bf i} - E_{\bf j} = E_{\bf k} - E_{\bf l}$ implies
\begin{align*}
n({\bf i}) + n({\bf l})-n({\bf j}) - n({\bf k}) = \frac{2m+1}{4} \big(s({\bf j}) + s({\bf k}) -s({\bf i})  - s({\bf l})\big),
\end{align*}
and thus, if the left-hand side is an integer number different from zero, it must be $|s({\bf j}) + s({\bf k}) -s({\bf i})  - s({\bf l})| = 4$, that is $s({\bf j}) = s({\bf k}) = -s({\bf i}) = - s({\bf l})$.

\medskip

The term  $E_{\bf j}^{(1)}$ coincides  with $V_{\bf j,j}=\langle \Phi_{\bf j}, (X\otimes \sigma_1) \Phi_{\bf j}\rangle$, thus we deduce from \eqref{interaction-matrix} that $E_{\bf j}^{(1)}=0$ for every ${\bf j}$.

\medskip

Following \cite{reed_simon} we have that 
$$E_{\bf j}^{(2)} = -\sum_{{\bf m} \neq {\bf j}} (E_{\bf m}-E_{\bf j})^{-1}V_{\bf j,m}V_{\bf m,j}.$$ Thus
\begin{align*}
E_{\bf j}^{(2)}&= -\sum_{{\bf m} \neq {\bf j}} (E_{\bf m}-E_{\bf j})^{-1}
(1-\delta_{s({\bf j}),s({\bf m})})^2 \left(\delta_{n({\bf j}),n({\bf m})-1}\sqrt{\frac{n({\bf j})+1}2}+\delta_{n({\bf j}),n({\bf m})+1}\sqrt{\frac{n({\bf j})}2}\right)^2\\
&= - (E_{n({\bf j})+1,-s({\bf j})}-E_{\bf j})^{-1}\,\frac{n({\bf j})+1}2 - (E_{n({\bf j})-1,-s({\bf j})}-E_{\bf j})^{-1}\,\frac{n({\bf j})}2\\
&= - (\omega-s({\bf j})\Omega)^{-1} \,\frac{n({\bf j})+1}2 + (\omega+s({\bf j})\Omega)^{-1}\,\frac{n({\bf j})}2\\
&= \frac{\omega +s({\bf j})\Om (2n({\bf j})+1)}{2(\Omega^2-\omega^2)}.
\end{align*}
Notice that the computation above is correct also for $n({\bf j})=0$, even if in this case 
$E_{n({\bf j})-1,-s({\bf j})}$ is not a well defined eigenvalue of $H_{\Rabi,0}$. Indeed, in this case the term 
$$(E_{n({\bf j})-1,-s({\bf j})}-E_{\bf j})^{-1}\frac{n({\bf j})}{2}$$ 
counts as zero.

Let us identify 
the values ${\bf i,j,k,l}$ such that $E_{\bf i}^{(2)}-E_{\bf j}^{(2)}=E_{\bf k}^{(2)}-E_{\bf l}^{(2)}$ under the assumption that $E_{\bf i} - E_{\bf j} = E_{\bf k} - E_{\bf l}$. Recall that we also assume that 
$({\bf i,j})\neq ({\bf k,l})$ and ${\bf i\neq j}$. 
From the above expression of $E_{\bf j}^{(2)}$ we have 
\begin{align}
s({\bf i}) \left(2 n({\bf i})+1\right) - s({\bf j}) \left(2n({\bf j})+1\right) =  s({\bf k}) \left(2n({\bf k})+1\right)-s({\bf l}) \left(2n({\bf l})+1\right).\label{rel-1}
 \end{align}
 If $\Omega=(2m+1)\omega/2$ for some non-negative integer $m$ and $s({\bf j}) = s({\bf k}) = -s({\bf i}) = - s({\bf l})$
then \eqref{rel-1} gives  $n({\bf i}) + n({\bf j}) +n({\bf k}) +n({\bf l}) + 2 =0$, which is impossible being the addends \gp{non-negative}.

The remaining case is when $n({\bf i})-n({\bf j}) = n({\bf k})-n({\bf l})$ and $s({\bf i})-s({\bf j}) = s({\bf k})-s({\bf l})$, in which case
\begin{align}
s({\bf i}) n({\bf i}) - s({\bf j}) n({\bf j}) = s({\bf k}) n({\bf k})-s({\bf l}) n({\bf l}).\label{rel}
\end{align}
Then, either $s({\bf i})=s({\bf j})$, which implies $s({\bf k})=s({\bf l})$ and then, by~\eqref{rel}, $s({\bf i})=s({\bf j})=s({\bf k})=s({\bf l})$, or $s({\bf i})=- s({\bf j})$, which implies $s({\bf k})=- s({\bf l})$ and then, by~\eqref{rel}, $n({\bf i})+n({\bf j})=n({\bf k})+n({\bf l})$. In the latter case it must be ${\bf i}={\bf k}$ and ${\bf j}={\bf l}$, which is excluded by assumption.
Therefore the nontrivial quadruples {\bf i},{\bf j},{\bf k},{\bf l} satisfying both the equalities $E_{\bf i} - E_{\bf j} = E_{\bf k} - E_{\bf l}$ and $E_{\bf i}^{(2)}-E_{\bf j}^{(2)}=E_{\bf k}^{(2)}-E_{\bf l}^{(2)}$
are those for which  $s({\bf i})=s({\bf j})=s({\bf k})=s({\bf l})$ and $n({\bf i})-n({\bf j}) = n({\bf k})-n({\bf l})$.

\medskip

Let us now evaluate the terms $E_{\bf j}^{(3)}$ as in \cite{reed_simon}. We have
\begin{align*}
E_{\bf j}^{(3)} &= \sum_{{\bf m} \neq {\bf j},{\bf p} \neq {\bf j}} (E_{\bf m}-E_{\bf j})^{-1}(E_{\bf p}-E_{\bf j})^{-1} V_{\bf j,m}V_{\bf m,p}V_{\bf p,j}-\sum_{\bf m\neq j} (E_{\bf m}-E_{\bf j})^{-2} V_{\bf j,m}V_{\bf m,j} V_{\bf j,j}.
\end{align*}
Since $V_{\bf a,b}\neq 0$ only if $s({\bf a}) = -s({\bf b})$ it turns out that $V_{\bf j,m}$ and $V_{\bf m,p}$ are different from $0$ only if $s({\bf j}) = s({\bf p}) =  -s({\bf m})$, but then $V_{\bf p,j}=0$. Thus, recalling that $V_{\bf j,j}=0$, we have
$E_{\bf j}^{(3)} = 0$ for every ${\bf j}\in\N\times\{-1,1\}$.

\medskip

We are going to complete the proof that
the set $S_{\bf i,j,k,l}$ defined as in
\eqref{ijkl} has full measure by showing that
if ${\bf i},{\bf j},{\bf k},{\bf l}\in \N\times\{-1,1\}$
are such that $({\bf i},{\bf j})\neq ({\bf k},{\bf l})$, ${\bf i\neq j}$, 
$n({\bf i})-n({\bf j}) = n({\bf k})-n({\bf l})$ (which follows from $E_{\bf i}-E_{\bf j}= E_{\bf k}-E_{\bf l}$)
 and $s({\bf i})=s({\bf j})=s({\bf k})=s({\bf l})$ (which follows from
$E_{\bf i}^{(2)}-E_{\bf j}^{(2)}= E_{\bf k}^{(2)}-E_{\bf l}^{(2)}$), then
$E_{\bf i}^{(4)}-E_{\bf j}^{(4)}\neq E_{\bf k}^{(4)}-E_{\bf l}^{(4)}$.

\medskip


The general formula for $E_{\bf j}^{(4)}$ (see \cite{reed_simon}) is 

\begin{align} \label{E_4}
E_{\bf j}^{(4)} &=  -\sum_{{\bf m} \neq {\bf j},{\bf p} \neq {\bf j},{\bf q} \neq {\bf j}} (E_{\bf m}-E_{\bf j})^{-1}(E_{\bf p}-E_{\bf j})^{-1} (E_{\bf q}-E_{\bf j})^{-1} V_{\bf j,m}V_{\bf m,p}V_{\bf p,q}V_{\bf q,j} \\ 
\nonumber
&\quad +\sum_{{\bf m} \neq {\bf j},{\bf p} \neq {\bf j}}\hspace{-2mm}V_{\bf j,j}V_{\bf j,m}V_{\bf m,p} V_{\bf p,j}
[(E_{\bf m}-E_{\bf j})^{-1}
(E_{\bf p}-E_{\bf j})^{-2}  + (E_{\bf m}-E_{\bf j})^{-2}(E_{\bf p}-E_{\bf j})^{-1}] \\ \nonumber
&\quad +\sum_{{\bf m} \neq {\bf j},{\bf p} \neq {\bf j}} (E_{\bf m}-E_{\bf j})^{-2}(E_{\bf p}-E_{\bf j})^{-1}
 V_{\bf j,m}V_{\bf m,j}V_{\bf j,p}V_{\bf p,j} \nonumber -\sum_{\bf m\neq j}  (E_{\bf m}-E_{\bf j})^{-3} V_{\bf j,m} V_{\bf m,j} V_{\bf j,j}^2.
\end{align}
Since $V_{\bf j,j}=0$, only the first and third term of the right-hand side must be evaluated.

Let us compute the first term {in \eqref{E_4}.} In order to avoid null terms we must assume $s({\bf j}) = -s({\bf m}) =s({\bf p})= -s({\bf q})$ and thus $n({\bf j}) \neq n({\bf p})$. Therefore the only nonzero terms in the sum are given by $n({\bf j}) = n({\bf m})+1 = n({\bf p})+2 =  n({\bf q})+1$ (if $n({\bf j})>1$) and $n({\bf j}) = n({\bf m})-1 = n({\bf p})-2 =  n({\bf q})-1$.
We have 
\begin{align*}
\lefteqn{\sum_{{\bf k} \neq {\bf j},{\bf l} \neq {\bf j},{\bf i} \neq {\bf j}} (E_{\bf k}-E_{\bf j})^{-1}(E_{\bf l}-E_{\bf j})^{-1} (E_{\bf i}-E_{\bf j})^{-1} 
 V_{\bf j,k}V_{\bf k,l}V_{\bf l,i}V_{\bf i,j}=}\\
 &\qquad(-\omega-s({\bf j})\Omega)^{-1} (-2\omega)^{-1}(-\omega-s({\bf j})\Omega)^{-1}\left(\frac{n({\bf j})}2\right)\left( \frac{n({\bf j})-1}2\right)\\
 &\qquad+(\omega-s({\bf j})\Omega)^{-1}(2\omega)^{-1}(\omega-s({\bf j})\Omega)^{-1} \left(\frac{n({\bf j})+1}2\right)\left( \frac{n({\bf j})+2}2\right).
\end{align*}

Notice that the formula is correct also in the case where $n({\bf j})=0$ or $n({\bf j})=1$.

Let us now compute the third term {in \eqref{E_4}.}
As before, to avoid null terms we assume $s({\bf j}) = -s({\bf m}) = - s({\bf p})$. The nonzero terms in the sum are given by $n({\bf m})=n({\bf j})\pm 1$ and  $n({\bf p})=n({\bf j})\pm 1$ thus we have to sum four terms.
We have
\begin{align*}
\sum_{{\bf m} \neq {\bf j},{\bf p} \neq {\bf j}} (E_{\bf m}-E_{\bf j})^{-2}(E_{\bf p}-E_{\bf j})^{-1} V_{\bf j,m}V_{\bf m,j}V_{\bf j,p}V_{\bf p,j}&=(-\omega-s({\bf j})\Omega)^{-2}(-\omega-s({\bf j})\Omega)^{-1}\left(\frac{n({\bf j})}2\right)^2\\
&+(-\omega-s({\bf j})\Omega)^{-2}(\omega-s({\bf j})\Omega)^{-1}\,\frac{n({\bf j})}2\, \frac{n({\bf j})+1}2\\
&+(\omega-s({\bf j})\Omega)^{-2}(-\omega-s({\bf j})\Omega)^{-1}\,\frac{n({\bf j})}2 \,\frac{n({\bf j})+1}2\\
&+(\omega-s({\bf j})\Omega)^{-2}(\omega-s({\bf j})\Omega)^{-1}\left(\frac{n({\bf j})+1}2\right)^2.
\end{align*}

By summing up all the terms one sees that, for fixed $s=s({\bf j})$, the term $E_{\bf j}^{(4)}$ depends quadratically on $n({\bf j})$, i.e. 
$$E_{\bf j}^{(4)} = C_0(s({\bf j}))+C_1(s({\bf j})) n({\bf j}) + C_2(s({\bf j})) n({\bf j})^2,$$
 where the coefficient $C_2(s({\bf j}))$ is given by
$$C_2(s({\bf j}))= s({\bf j}) \frac{\Omega(\omega^2+3 \Omega^2)}{2 (\omega^2- \Omega^2 )^3} \neq 0.
$$

So, {if ${\mathbf{i}, \mathbf{j}, \mathbf{k}, \mathbf{l}}$} are such that $s({\bf i})=s({\bf j})=s({\bf k})=s({\bf l})=s$ and $n({\bf i})-n({\bf j}) = n({\bf k})-n({\bf l})$, we have
\begin{align*}
 E_{\bf i}^{(4)}- E_{\bf j}^{(4)}=E_{\bf k}^{(4)}- E_{\bf l}^{(4)} 
\iff & C_1(s) (n({\bf i})- n({\bf j})) + C_2(s) (n({\bf i})^2 -n({\bf j})^2) = \\
&\quad C_1(s) (n({\bf k})- n({\bf l})) + C_2(s) (n({\bf k})^2 -n({\bf l})^2) \\
\iff &  C_2(s) (n({\bf i})^2 -n({\bf j})^2) =  C_2(s) (n({\bf k})^2 -n({\bf l})^2) \\
\iff &  C_2(s) (n({\bf i}) +n({\bf j})) =  C_2(s) (n({\bf k}) + n({\bf l}) )\\
\iff &  n({\bf i}) =  n({\bf k}) \mbox{ and } n({\bf j}) =  n({\bf l}).
 \end{align*}

This concludes the proof that
for almost every $g\in \R$ and every ${\bf i,j,k,l}\in \N\times\{-1,1\}$, with $({\bf i,j})\neq ({\bf k,l})$ and ${\bf i\neq j}$, 
one has $E_{\bf i}^g-E_{\bf j}^g\neq E_{\bf k}^g-E_{\bf l}^g$.


\subsubsection{Step 2: Coupling of the relevant energy levels.}
\label{Sec_coupling}

The proof of Theorem~\ref{Tmain} is then concluded, thanks to Theorem~\ref{pro_approx_contr_norme_H},
if we show that  the controlled Hamiltonian $x\otimes \Id$ couples, directly or indirectly,  all the energy levels for almost all $g\in \R$.

More precisely, we show below that $\langle \Phi_{\bf j}^g,(x\otimes \Id)\Phi_{\bf k}^g\rangle \neq 0$ for almost every $g\in\R$ for all ${\bf j,k}$ such that $s({\bf j})=s({\bf k})$ and $| n({\bf j}) - n({\bf k}) | =1$ or $s({\bf j})=-s({\bf k})$ and $n({\bf j}) = n({\bf k})$.  See Figure \ref{e-vals}. As before, it is enough to show that the corresponding Taylor series in $g$ is nonzero.

\begin{figure}
\begin{center}
\includegraphics[width=3.7cm]{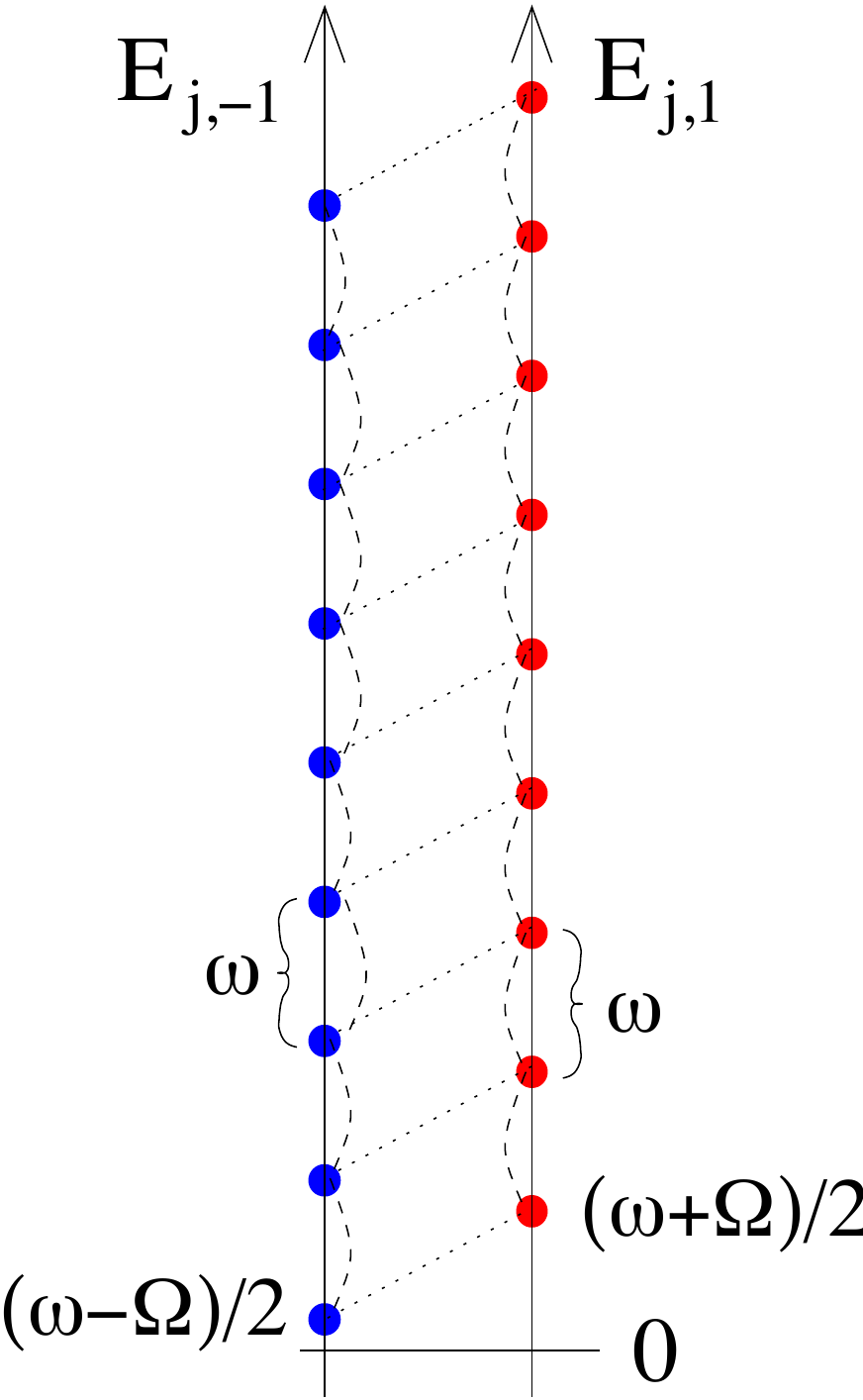}
\caption{\label{e-vals} The dashed lines connect eigenvalues of $H_{\Rabi}$ that are
``coupled'' by the controlled Hamiltonian when $g=0$, while the dotted
lines connect eigenvalues that are coupled by the controlled
Hamiltonian for almost all $g\neq 0$.} 
\end{center}
\end{figure}

Set $\Phi_{\bf j}^g=\Phi_{\bf j}+\sum_{m=1}^{\infty} g^m \Phi_{\bf j}^{(m)}$.
We have 
\begin{align*}
\langle \Phi_{\bf j},(x\otimes \Id)\Phi_{\bf k}\rangle =\left(\delta_{n({\bf j}),n({\bf k})-1}\sqrt{\frac{n({\bf j})+1}2}+\delta_{n({\bf j}),n({\bf k})+1}\sqrt{\frac{n({\bf j})}2}\right) \delta_{s({\bf j}),s({\bf k})}.
\end{align*}
This is enough to say that $\langle \Phi_{\bf j}^g,(x\otimes \Id)\Phi_{\bf k}^g\rangle \neq 0$ for almost every $g$ for all ${\bf j,k}$ such that $s({\bf j})=s({\bf k})$ and $| n({\bf j}) - n({\bf k}) | =1$.

The term $\Phi_{\bf j}^{(1)}$ can be characterized through the relation
\begin{align}
(H_{\Rabi,0}+g x\otimes \sigma_1)(\Phi_{\bf j}+g \Phi_{\bf j}^{(1)}+o(g))=(E_{\bf j}+g E_{\bf j}^{(1)}+o(g))(\Phi_{\bf j}+g \Phi_{\bf j}^{(1)}+o(g)).\label{frod}
\end{align}
Regrouping the first-order terms in \eqref{frod} we get
\begin{equation}\label{o-1}
H_{\Rabi,0} \Phi_{\bf j}^{(1)} +(x\otimes \sigma_1)\,\Phi_{\bf j}-
E_{\bf j}\Phi_{\bf j}^{(1)} -E_{\bf j}^{(1)}\Phi_{\bf j}=0.
\end{equation}
Denote by $\Pi$ the orthogonal projection on the orthogonal complement to $\Phi_{\bf j}$.
Applying $\Pi$ to \eqref{o-1}, we get
$$(H_{\Rabi,0}-E_{\bf j}\Id)\Phi_{\bf j}^{(1)}+ \Pi(x\otimes \sigma_1
)\Phi_{\bf j} = 0.$$
Notice that
the orthogonal complement to $\Phi_{\bf j}$ is an invariant space for the operator
$H_{\Rabi,0}-E_{\bf j}\Id$, which is invertible when restricted to it. We write
$(H_{\Rabi,0}-E_{\bf j}\Id)^{-1}$ to denote its inverse (whose values are in the orthogonal complement
to $\Phi_{\bf j}$).
Thus,
\begin{align*}
\Phi_{\bf j}^{(1)}&= -(H_{\Rabi,0}-E_{\bf j}\Id)^{-1}
 \Pi(x\otimes \sigma_1
 )\Phi_{\bf j}= \sum_{\bf l\neq j} (E_{\bf j}-E_{\bf l})^{-1} \langle  \Phi_{\bf l}, (x\otimes \sigma_1) \Phi_{\bf j}\rangle \Phi_{\bf l}.
\end{align*}

The linear term in the Taylor expansion of $\langle \Phi_{\bf j}^g,(x\otimes \Id)\Phi_{\bf k}^g\rangle$ with respect to $g$ is given by
\begin{align*}
\langle \Phi_{\bf j},(x\otimes \Id)\Phi_{\bf k}^{(1)}\rangle+\langle \Phi_{\bf j}^{(1)},(x\otimes \Id)\Phi_{\bf k}\rangle=&\sum_{\bf l\neq k} (E_{\bf k}-E_{\bf l})^{-1} \langle  \Phi_{\bf l}, (x\otimes \sigma_1) \Phi_{\bf k}\rangle \langle \Phi_{\bf j}, (x\otimes \Id)\Phi_{\bf l}\rangle+\\
&\sum_{\bf l\neq j} (E_{\bf j}-E_{\bf l})^{-1} \langle  \Phi_{\bf l}, (x\otimes \sigma_1) \Phi_{\bf j}\rangle \langle \Phi_{\bf l}, (x\otimes \Id)\Phi_{\bf k}\rangle.
\end{align*}

Taking into account only the nonzero terms in the first sum gives $s({\bf j}) = s({\bf l}) = -s({\bf k})$, $|n({\bf j})-n({\bf l})|=1$ and $|n({\bf k})-n({\bf l})|=1$, so for fixed ${\bf j,k}$ we only have (at most) two terms. We assume $n({\bf j})=n({\bf k})$, thus we have
\begin{align*}
\langle \Phi_{\bf j},(x\otimes \Id)\Phi_{\bf k}^{(1)}\rangle=&\left[(-\omega+s({\bf k})\Omega)^{-1}\frac{n({\bf k})+1}2 + (\omega+s({\bf k})\Omega)^{-1}\frac{n({\bf k})}2  \right]\\
=&\frac{\omega+s({\bf k})\Omega (1+2  n({\bf k}))}{2(\Omega^2-\omega^2)}.
\end{align*}
Similarly,
\begin{align*}
{\langle \Phi_{\bf j}^{(1)},(x\otimes \Id)\Phi_{\bf k}\rangle=}
&\frac{\omega-s({\bf k})\Omega (1+2  n({\bf k}))}{2(\Omega^2-\omega^2)}.
\end{align*}

Hence,
\begin{align*}
\langle \Phi_{\bf j}^{(1)},(x\otimes \Id)\Phi_{\bf k}\rangle+{\langle \Phi_{\bf j},(x\otimes \Id)\Phi_{\bf k}^{(1)}\rangle}
&=\frac{\omega}{\Omega^2-\omega^2},
\end{align*}
which is different from zero for every {\bf k}. Thus the Taylor series of $\langle \Phi_{\bf j}^g,(x\otimes \Id)\Phi_{\bf k}^g\rangle $ is nonzero, which concludes the proof of the theorem in the case $\omega\ne\Omega$.

\subsection{Case $\omega=\Omega$}\label{case=}

Let us assume in this section that $\omega=\Omega$. Hence 
$$E_{j,1}=E_{j+1,-1}=\omega(j+1)$$ 
for every $j\in\N$ and $E_{0,-1}=0$.

Recall that, as $g$ varies in $\R$,  a complete set of eigenpairs for $H_\Rabi$ can be given in the form
$(E_{\bf j}^g,\Phi_{\bf j}^g)_{{\bf j}\in\N\times\{-1,1\}}$
with each $E_{\bf j}^g$ and $\Phi_{\bf j}^g$ depending analytically on $g$ and $E_{\bf j}^0=E_{\bf j}$. 
Notice that each $\Phi_{\bf j}^g$ can be taken in $L^2(\R,\R)\otimes \C^2$ and that we can take 
$\Phi_{0,-1}^0=\Phi_{0,-1}=\varphi_0\otimes \nu_{-1}$. 

The eigenfunctions  $\Phi_{j,1}^0,\Phi_{j+1,-1}^0$ are characterized by the relation
$$ \langle \Phi_{j,1}^0,(X\otimes \sigma_1)\Phi_{j+1,-1}^0\rangle=0.$$
(This follows by standard perturbation relations
recalled in the Appendix, 
 see equation \eqref{splitting}.)

Since, moreover, $\Phi_{j,1}^0$ and $\Phi_{j+1,-1}^0$ are orthogonal and of norm one, one gets, up to signs,
$$ \{\Phi_{j,1}^0,\Phi_{j+1,-1}^0\}=\left\{\frac{\Phi_{j,1}+\Phi_{j+1,-1}}{\sqrt{2}},\frac{\Phi_{j,1}-\Phi_{j+1,-1}}{\sqrt{2}}\right\}.$$

Let us fix, by convention, 
$$ \Phi_{j,1}^0=\frac{\Phi_{j,1}+\Phi_{j+1,-1}}{\sqrt{2}},\qquad \Phi_{j+1,-1}^0=\frac{\Phi_{j,1}-\Phi_{j+1,-1}}{\sqrt{2}}.$$

With this choice \eqref{splitting} yields
$$\frac{d}{dg}E_{j,1}^g\big|_{g=0}=\sqrt{\frac{j+1}{2}},\qquad \frac{d}{dg}E_{j+1,-1}^g\big|_{g=0}=-\sqrt{\frac{j+1}{2}}.$$

Assume now that
\begin{align}
E_{{\bf j}_1}^0-E_{{\bf j}_2}^0&=E_{{\bf j}_3}^0-E_{{\bf j}_4}^0,\label{o1}\\
\frac{d}{dg}\big|_{g=0} (E_{{\bf j}_1}^g-E_{{\bf j}_2}^g)&=\frac{d}{dg}\big|_{g=0}(E_{{\bf j}_3}^g-E_{{\bf j}_4}^g),\label{o2}
\end{align}
and let us prove that ${\bf j}_1={\bf j}_3$ and ${\bf j}_2={\bf j}_4$.
Write $E_{{\bf j}_l}^0=\omega k_l $ for $l=1,\dots,4$ for some $k_1,\dots,k_4\in \N$. 
Then \eqref{o1} implies that $k_1=k_2+p$ and $k_3=k_4+p$ for some $p\in \Z$.

Assume that $k_1,\dots,k_4\ne 0$ (the other cases being similar). 
Equality \eqref{o2} then yields
\begin{equation}\label{s14}
 s_1 \sqrt{k_2+p}-s_2 \sqrt{k_2}=s_3 \sqrt{k_4+p}-s_4 \sqrt{k_4}
 \end{equation}
for some $s_1,\dots,s_4\in\{-1,1\}$
and we are left to prove that $s_1=s_3$, $s_2=s_4$ and $k_2=k_4$. Without loss of generality we can assume that $k_2+p=\max\{k_2,k_2+p,k_4,k_4+p\}$ and that $s_1=1$. Then $p$ and 
the left-hand side in \eqref{s14} are nonnegative, which implies that $s_3=1$. 
We can then rewrite \eqref{s14} as
\begin{equation}\label{s14bis}
\sqrt{k_2+p}-s_2 \sqrt{k_2}=\sqrt{k_4+p}-s_4 \sqrt{k_4}
 \end{equation}
 If $s_2=s_4$ then, by monotonicity of $k\mapsto \sqrt{k+p}\pm \sqrt{k}$ we deduce that $k_2=k_4$ and we are done. 
Finally, if $s_2=-s_4$ then either $s_2=-1$, and then $k_4> k_2$, or $s_2=1$ and then $\sqrt{k_2+p}=\sqrt{k_4+p}+ \sqrt{k_4}+ \sqrt{k_2}>\sqrt{p}+\sqrt{k_2}\geq \sqrt{k_2+p}$, leading in both cases to a contradiction.

In order to apply Theorem~\ref{pro_approx_contr_norme_H} we are left to prove that 
 the controlled Hamiltonian $X\otimes \Id$ couples, directly or indirectly, 
 all the elements of the basis $(\Phi_{\bf j}^0)_{{\bf j}\in \N\times\{-1,1\}}$. 
 
First notice that
$$\langle \Phi_{j,1}^0,(X\otimes \Id)\Phi_{j+1,1}^0\rangle=\frac12\left(\sqrt{\frac{j+1}2}+\sqrt{\frac{j+2}2}\right)\ne 0,\qquad j\in \N,$$
meaning that  $X\otimes \Id$ couples all basis elements of the type 
$\Phi_{j,1}^0$. 

Similarly, one has that $\langle \Phi_{j,-1}^0,(X\otimes \Id)\Phi_{j+1,-1}^0\rangle\ne 0$ for every $j\in \N$. 
Finally, one easily checks that 
$\langle \Phi_{j,1}^0,(X\otimes \Id)\Phi_{j,-1}^0\rangle\ne 0$ for every $j\in \N$, completing the proof.

\section*{Appendix: global real analyticity of eigenpairs}

In this section we detail the proof of the analyticity of the eigenpairs $(E_{\bf j}^{g},\Phi_{\bf j}^{g})$ as functions of $g\in\R$. 
This result is somehow {\it folklore} and follows from well-known theorems in linear perturbation theory (\cite{Kato,reed_simon}). 
Since the latter are spread over the literature, we prefer to sketch here the main ideas of the argument, which is also useful in other contexts, as e.g. for the degenerate Gell-Man and Low theorem (\cite{Brouder2009a,Brouder2009,Brouder2008}).

Recall that, 
given two self-adjoint operators $H_0,W$  on a separable Hilbert space,
 $W$ is said to be {\it Kato-small with respect to $H_0$} if $D(H_0)\subset D(W)$ and for every $a>0$ there exists $b>0$ such that 
$\|W\psi\|\leq a \|H_0 \psi\| + b\|\psi\|$ for every $\psi\in D(H_0)$. 
According to this definition, it is easy to check that  $V$, defined as in \eqref{RabiV}, is Kato-small with respect to $H_{\Rabi,0}$.

\begin{prop}
Let $H_0,W$ be self-adjoint operators on a separable Hilbert space $\mathcal{H}$. 
Assume that $W$ is Kato-small with respect to $H_0$ and that $H_0$ has compact resolvent. 
Then 
$H_g=H_0+ g W$ 
is self-adjoint with compact resolvent for every $g\in \R$, hence it admits a complete orthonormal system of eigenfunctions $\{\psi_j(g)\}_{j\in\N}$ and the corresponding eigenvalues $\{E_j(g)\}_{j\in\N}$ have finite multiplicity and do not accumulate at any finite point. 
Moreover, up to a suitable choice of labelling, for each $j\in\N$ the function $g\mapsto  (E_j(g),\psi_j(g))$ is analytic from $\R$ to $\R\times \mathcal{H}$. 
\end{prop}
\begin{proof}
 Self-adjointness of $H_g$ follows from \cite[Theorem X.12]{reed-simon-2}.
Let us first prove that  $H_g$  has compact resolvent for every $g\in\R$ \gp{(we provide a direct proof; see also 
\cite[Theorem VII.2.4]{Kato})}. 
Since\footnote{Notice that $(H_g-\iu \Id )^{-1}$ maps $\Hi$ into $D(H_g) \subset D(W)$, so that 
$W (H_g-\iu \Id )^{-1}$ is well-defined on the whole Hilbert space $\Hi$.}
$$
(H_g-\iu \Id)^{-1}=(H_0-\iu \Id)^{-1} \left(\Id - g  \, W (H_g-\iu \Id)^{-1} \right)
$$
and $(H_0- \iu \Id)^{-1}$ is compact, we are left to prove that $W(H_g- \iu \Id)^{-1}$ is bounded.
Notice that for every $a>0$ there exists $b>0$ such that
$\|W\psi\|\leq a \|(H_g-\iu \Id) \psi \| + b\|\psi\|$ for every $\psi\in D(H_0)=D(H_g-\iu \Id )$. 
Setting $ \psi=(H_g-\iu \Id)^{-1}\phi$ for $\phi \in \Hi$, we get
$$\|W(H_g- \iu \Id)^{-1} \phi \|\leq a \|\phi \| + b\|(H_g-\iu \Id)^{-1} \phi \|\leq (a+b)\|\phi\|,$$ 
where the last inequality follows from $\|(H_g-\iu \Id)\eta \|\|\eta \|\geq |\langle (H_g-\iu \Id ) \eta,\eta \rangle|\geq \|\eta \|^2$, 
for $\eta \in D(H_g - \iu \Id)$.

If $E_{j}(g)$ is simple then,
by Rellich's theorem \cite[Theorem XII.8]{reed_simon}, 
 up to a reordering of the spectrum of $H_{g'}$ for $g'$ in a neighborhood of $g$, 
the eigenpair parameterization $g'\mapsto (E_{j}(g'),\psi_{j}(g'))$ is analytic near $g$.  
The previous result can be generalized to the case of eigenvalues of multiplicity $m\geq 2$. 
Indeed, if $E_{j_1}(g)=\dots=E_{j_m}(g)$ and ${j}_1,\dots,{j}_m$ are distinct, 
we can assume, up to relabelling, that $g'\mapsto (E_{{j}_1}(g'),\psi_{{j}_1}(g')),\dots,g'\mapsto (E_{{j}_m}(g'),\psi_{{j}_m}(g'))$ are analytic near $g$
(\cite[Theorem XII.13]{reed_simon}). Moreover, 
\begin{equation}\label{splitting}
\langle \psi_{{ j}_l}(g), W\psi_{{j}_k}(g)\rangle=\delta_{lk}\frac {d}{d g}E_{{j}_l}(g),\quad l,k=1,\dots,m.
\end{equation}
(See, for instance, \cite{albert}.)

In order to describe the spectrum by a countable family of functions which are globally analytic on  $\R$,
we should ensure that all such locally analytic functions $E_j(\cdot)$ can be extended indefinitely. 

Fix $j$ and let $I$ be the maximal interval containing $0$ such that, up to relabelling, $E_j$ is analytic on $I$. Let us prove that, for every  $\ell \in (0,+\infty)$, $ [-\ell,\ell]\subset I$.

Set $a=(2\ell)^{-1}$. Take $b>0$ such that 
 $\|W\psi\|\leq a \|H_0 \psi\| + b\|\psi\|$ for every $\psi\in D(H_0)$. 
 Hence, for $\psi\in D(H_0)$ and $g\in [-\ell,\ell]$,
\begin{align*}
 \|W\psi\|&\le a \|(H_0+gW)\psi\|+a|g|\|W\psi\|+b \|\psi\|\\
&\le  a \|H_g\psi\|+\frac{\|W\psi\|}2+b \|\psi\|
 \end{align*}
 and therefore
\begin{align*}
 \|W\psi\|&\le 2 a \|H_g\psi\|+2b \|\psi\|,\quad g\in[-\ell,\ell].
 \end{align*}
 
It follows from \eqref{splitting} that, for every $g\in I\cap [-\ell,\ell]$, 
 $$\left|\frac{d}{dg}E_j(g)\right|=|\langle \psi_j(g),W\psi_j(g)\rangle|\leq \|W \psi_j(g)\|\leq 2a  \|H_g\psi_j(g)\|+2b \|\psi_j(g)\|.
$$
Hence,
 $$\left|\frac{d}{dg}|E_j(g)|\right|\leq 2a  |E_j(g)|+2b,\qquad \mbox{ for almost every $g\in I\cap [-\ell,\ell]$.}
$$

A standard application of Gronwall inequality yields that $E_j(\cdot)$ is Lipschitz on $I\cap [-\ell,\ell]$.  In particular the limits of $E_j(\cdot)$ at the boundary of $I\cap[-\ell,\ell]$ are well defined and they are eigenvalues of the corresponding perturbed operators (see e.g.~\cite[Theorem XII.7]{reed_simon}). Therefore $E_j(\cdot)$ can be extended indefinitely on $[-\ell,\ell]$ and thus also on the whole real line.  
\end{proof}

\section*{Acknowledgements}
We are grateful to L. Pinna for a careful reading of the manuscript and some useful comments, and to the {\it  Institut Henri Poincar\'e} for the kind hospitality in the framework of the trimester \virg{Variational and Spectral Methods in Quantum Mechanics}, organized by M.\,J.\ Esteban and M.\ Lewin. 

\medskip

\noindent This research has been supported by the European Research Council, ERC StG 2009 ``GeCoMethods'', contract number 239748,  by the iCODE institute, research project of the Idex Paris-Saclay and by the EU FP7 project QUAINT, grant agreement no. 297861.

{
\bibliographystyle{abbrv}
\bibliography{biblioRABI}
}

\end{document}